\documentclass[journal,draftcls,onecolumn,12pt,twoside]{IEEEtran}

\usepackage[T1]{fontenc} % optional
\usepackage{times}
\usepackage{amsmath}
\usepackage{amssymb}
\usepackage{amsthm}
\usepackage{color}
\usepackage{algorithm}
\usepackage[noend]{algorithmic}
\usepackage{graphicx}
\usepackage{subfigure}
\usepackage{multirow}
\usepackage[bookmarks=false,colorlinks=false,pdfborder={0 0 0}]{hyperref}
\usepackage{cite}
\usepackage{bm}
\usepackage{arydshln}
\usepackage{mathtools}
\usepackage{microtype}
\usepackage{subfigure}
\usepackage{float}

\usepackage[figuresright]{rotating}
\usepackage{threeparttable}
\usepackage{booktabs}
\usepackage{color}

\newtheorem{theorem}{Theorem}

\newtheorem{lemma}[theorem]{Lemma}
\newtheorem{corollary}[theorem]{Corollary}

\long\def\symbolfootnote[#1]#2{\begingroup
\def\thefootnote{\fnsymbol{footnote}}\footnote[#1]{#2}\endgroup}
\renewcommand{\paragraph}[1]{{\bf #1}}

\long\def\symbolfootnote[#1]#2{\begingroup
\def\thefootnote{\fnsymbol{footnote}}\footnote[#1]{#2}\endgroup}

\ifodd1%revise of the text
\newcommand{\com}[1]{\textbf{\color{blue} (COMMENT: #1)}}%comment of the text
\else\newcommand{\com}[1]{}\fi

\begin{document}
\title{Two New Piggybacking Designs with Lower Repair Bandwidth}

\author{Zhengyi Jiang, Hanxu Hou, Yunghsiang S. Han,
Patrick P. C. Lee, Bo Bai, and Zhongyi Huang
}

\maketitle

\begin{abstract}\symbolfootnote[0]{
Zhengyi Jiang and Zhongyi Huang are with the Department of Mathematics Sciences, Tsinghua University
(E-mail: jzy21@mails.tsinghua.edu.cn, zhongyih@tsinghua.edu.cn).
Hanxu Hou and Bo Bai are with Theory Lab, Central Research Institute, 2012 Labs, Huawei Technology Co. Ltd.
(E-mail: houhanxu@163.com, baibo8@huawei.com).
Yunghsiang S. Han is with the Shenzhen Institute for Advanced Study, University of Electronic Science and Technology of China~(E-mail: yunghsiang@gmail.com).
Patrick P. C. Lee is with the Department of Computer Science and Engineering,
The Chinese University of Hong Kong (E-mail: pclee@cse.cuhk.edu.hk). This
work was partially supported by the National Key R\&D Program of China (No.
2020YFA0712300), National Natural Science Foundation of China (No. 62071121, No.12025104, No.11871298),
Research Grants Council of HKSAR (AoE/P-404/18), Innovation and Technology
Fund (ITS/315/18FX).
}
Piggybacking codes are a special class of MDS array codes that can achieve small repair
bandwidth with small sub-packetization by first creating some instances of an $(n,k)$ MDS code,
such as a Reed-Solomon (RS) code, and then designing the
piggyback function. In this paper, we propose a new piggybacking coding design
which designs the piggyback function over some instances of both $(n,k)$ MDS code
and $(n,k')$ MDS code, when $k\geq k'$. We show that our new piggybacking
design can significantly reduce
the repair bandwidth for single-node failures.
When $k=k'$, we design a piggybacking code
that is MDS code and we show that the designed code has lower repair bandwidth
for single-node failures than all existing piggybacking codes when the number
of parity node $r=n-k\geq8$ and the sub-packetization $\alpha<r$.

Moreover, we propose another piggybacking codes by designing $n$ piggyback functions of
some instances of $(n,k)$ MDS code and adding the $n$ piggyback functions into the $n$
newly created empty entries with no data symbols. We show that our code can significantly
reduce repair bandwidth for single-node failures at a cost of slightly more storage
overhead. In addition, we show that our code
can recover any $r+1$ node failures for some parameters. We also show that
our code  has lower repair bandwidth than
locally repairable codes (LRCs) under the same fault-tolerance and redundancy
for some parameters.
\end{abstract}


\IEEEpeerreviewmaketitle

\section{Introduction}
\label{sec:intro}
{\em Maximum distance separable (MDS)} array codes are widely employed in distributed
storage systems that can provide the maximum data reliability for a given amount of
storage overhead.
An $(n,k,\alpha)$ MDS array code encodes a data file of $k\alpha$ {\em data symbols} to
obtain $n\alpha$ {\em coded symbols} with each of the $n$ nodes
storing $\alpha$ symbols such that any $k$ out of $n$ nodes can retrieve all
$k\alpha$ data symbols, where $k < n$ and $\alpha\geq 1$. The number of symbols stored in
each node, i.e., the size of $\alpha$, is called {\em sub-packetization level}.
We usually employ \emph{systematic code} in practical storage systems
such that the $k\alpha$ data symbols are directly stored in the system and can be retrieve without performing any decoding operation. Note that
Reed-Solomon (RS) codes \cite{reed1960} are typical MDS codes with $\alpha=1$.

In modern distributed storage systems, node failures are common and single-node failures
occur more frequently than multi-node failures \cite{ford2010}. When a single-node fails,
it is important to repair the failed node with the {\em repair bandwidth}
(i.e,. the total amount of symbols downloaded from other surviving nodes) as small
as possible. It is shown in \cite{dimakis2010} that we need to download at least
$\frac{\alpha}{n-k}$ symbols from each of the $n-1$ surviving nodes in
repairing one single-node failure.
MDS array codes with minimum repair bandwidth for any single-node failure are called
{\em minimum storage regenerating} (MSR) codes.
There are many constructions of MSR codes to achieve minimum repair bandwidth in the literature
\cite{rashmi2011,tamo2013,hou2016,2017Explicit,li2018,2018A,hou2019a,hou2019b}.
However, the sub-packetization level $\alpha$ of high-code-rate (i.e., $\frac{k}{n}>0.5$) MSR
codes \cite{2018A} is exponential in parameters $n$ and $k$. A nature question is that
can we design new MDS array codes with both sub-packetization and repair bandwidth
as small as possible.

Piggybacking codes \cite{2014A,2017Piggybacking} are a special class of MDS array codes that have
small sub-packetization and small repair bandwidth. The essential idea behind the piggybacking
codes \cite{2017Piggybacking} is as follows: by creating $\alpha$ instances of $(n,k)$ RS codes and
adding carefully well-designed linear combinations of some symbols as so-called piggyback functions from one instance to
the others, we can reduce the repair bandwidth of single-node failure.
Some further studies of piggybacking codes are in \cite{2014Sytematic,2018Repair,2019AnEfficient,2016A,2021piggyback,2021piggybacking}.

The existing piggybacking codes are designed based on some instances of
an $(n,k)$ RS codes.
The motivation of this paper is to significantly reduce the repair bandwidth by designing
new piggybacking codes.
In this paper, we propose new piggybacking codes by first creating some
instances of both $(n,k)$ MDS code and $(n,k')$ MDS code, and then designing the piggyback functions
that can significantly reduce repair bandwidth for single-node failures, when $k\geq k'$.

\subsection{Contributions}
Our main contributions are as follows.
\begin{itemize}
	\item First, we propose a new type of piggybacking coding design which is designed by
both $(n,k)$ MDS code and $(n,k')$ MDS code, where $k\geq k'$.
We give an efficient repair method for any single-node failure for our piggybacking
coding design and present an upper bound on repair bandwidth.
When $k>k'$, our codes are non-MDS codes and we show that our codes have much less repair bandwidth
than that of existing piggybacking codes at a cost of slightly more storage overhead.
The essential reason of repair bandwidth reduction of our codes is that we have more
design space than that of existing piggybacking codes.
	\item Second, when $k=k'$, we design new piggybacking codes that are MDS codes
based on the proposed design. We show that the proposed  piggybacking codes with $k=k'$
have lower repair bandwidth than that of the existing piggybacking codes when
$r=n-k\geq 8$ and the sub-packetization is less than $r$.
	\item Third, we design another piggybacking codes by designing and adding the
$n$ piggyback functions into the $n$ newly created empty entries with no data symbols.
We show that our piggybacking codes
can tolerant any $r+1$ node failures under some conditions. We also show that
our codes have lower repair bandwidth than that of both Azure-LRC \cite{huang2012}
and optimal-LRC \cite{2014optimal} under the same fault-tolerance and
the same storage overhead for some parameters.
\end{itemize}

\subsection{Related Works}
Many works are designed to reduce the repair bandwidth of erasure codes which we discuss as follows.

\subsubsection{Piggybacking Codes}
Rashmi \emph{et al.} present the seminal work of piggybacking codes \cite{2014A,2017Piggybacking} that
can reduce the repair bandwidth for any single-data-node with small sub-packetization.
Another piggybacking codes called REPB are proposed \cite{2018Repair} to achieve
lower repair bandwidth for any single-data-node than that of the codes in \cite{2017Piggybacking}.
Note that the piggybacking codes in \cite{2017Piggybacking,2018Repair} only have
small repair bandwidth for any single-data-node failure, while not for parity nodes.
Some follow-up works \cite{2019AnEfficient,2021piggyback,2021piggybacking} design new piggybacking codes
to obtain small repair bandwidth for both data nodes and parity nodes.
Specifically, when $r=n-k\leq10$ and sub-packetization is $r-1+\sqrt{r-1}$, OOP codes
\cite{2019AnEfficient} have the lowest repair bandwidth for any single-node failure
among the existing piggybacking codes; when $r\geq10$ and sub-packetization is $r$,
the codes in \cite{2021piggybacking} have the lowest repair bandwidth for any
single-node failure among the existing piggybacking codes.

Note that all the existing piggybacking codes are designed over some instances of
an $(n,k)$ MDS code.
In this paper, we design new piggybacking codes that are non-MDS codes over some
instances of both $(n,k)$ MDS code and $(n,k')$ MDS codes with $k>k'$
that have much lower repair bandwidth for any single-node failures at a cost of slightly
larger storage overhead.

\subsubsection{MDS Array Codes}
Minimum storage regenerating (MSR) codes \cite{dimakis2010} are a class of MDS array codes
with minimum repair bandwidth for a single-node failure. Some exact-repair
constructions of MSR codes are investigated in
\cite{rashmi2011,shah2012,tamo2013,hou2016,ye2017,li2018,hou2019a,hou2019b}.
The sub-packetization of high-code-rate MSR codes
\cite{tamo2013,ye2017,li2018,hou2019a,hou2019b} is exponentially increasing
with the increasing of parameters $n$ and $k$. Some MDS array codes have been proposed
\cite{corbett2004row,blaum1995evenodd,Hou2018A,xu1999x,2018MDS,2021A} to achieve small repair
bandwidth under the condition of small sub-packetization; however, they either
only have small repair bandwidth for data nodes
\cite{corbett2004row,blaum1995evenodd,hou2018d,Hou2018A,xu1999x} or require large field sizes
\cite{2018MDS,2021A}.

\subsubsection{Locally Repairable Codes}
Locally repairable codes (LRCs) \cite{huang2012,2014Locally} are non-MDS codes that can achieve
small repair bandwidth for any single-node failure with sub-packetization
$\alpha=1$ by adding some local parity symbols. Consider the $(n,k,g)$
Azure-LRC \cite{huang2012} that is employed in Windows Azure storage systems,
we first create $n-k-g$ global parity symbols by encoding
all $k$ data symbols, divide the $k$ data symbols into $g$
groups and then create one local parity symbol for each group,
where $k$ is a multiple of $g$. In the $(n,k,g)$
Azure-LRC, we can repair any one symbol except $n-k-g$ global parity symbols
by locally downloading the other $k/g$ symbols in the group. Optimal-LRC \cite{2014optimal,2019How,2020Improved,2020On}
is another family of LRC that can locally repair any one symbol (including the global
parity symbols). One drawback of optimal-LRC is that existing constructions
\cite{2014optimal,2019How,2020Improved,2020On} can not support all the
parameters and the underlying field size should be large enough.
In this paper, we propose new
piggybacking codes by designing and adding the
$n$ piggyback functions into the $n$ newly created empty entries with no data symbols
that are also non-MDS codes and we
show that our piggybacking codes have lower repair
bandwidth when compared with Azure-LRC \cite{huang2012} and optimal-LRC under the same storage overhead and
fault-tolerance, for some parameters.

%\subsection{Paper Organization}

The remainder of this paper is organized as follows. Section \ref{sec:2}
presents two piggybacking coding designs. Section \ref{sec:3} shows new
piggybacking codes with $k=k'$ based on the first design. Section \ref{sec:4} shows another new piggybacking
codes based on the second design. Section \ref{sec:com} evaluates the repair bandwidth for our
piggybacking codes and the related codes. Section \ref{sec:con} concludes the paper.

\section{Two Piggybacking Designs}
\label{sec:2}
In this section, we first present two piggybacking designs
and then consider the repair bandwidth of any single-node failure for
the proposed piggybacking codes.

\subsection{Two Piggybacking Designs}
\label{sec:2.1}
Our two piggybacking designs can be represented by an $n\times (s+1)$
array, where $s$ is a positive integer, the $s+1$ symbols in each row are stored in a node, and $s+1\le n$.
We label the index of the $n$ rows from 1 to $n$ and the index of the
$s+1$ columns from 1 to $s+1$. Note that the symbols in each row are stored at the corresponding node.

In the following, we present our first piggybacking design.
In the piggybacking design, we first create $s$ instances of $(n,k)$ MDS codes
plus one instance of $(n,k')$ MDS codes and then design the piggyback functions,
where $k\geq k'>0$.
We describe the detailed structure of the design as follows.

\begin{enumerate}[]
\item First, we create $s+1$ instances of MDS codes over finite field $\mathbb{F}_q$, the first $s$
columns are the codewords of $(n,k)$ MDS codes and the last column is
a codeword of $(n,k')$ MDS codes, where $k'=k-h$, $h\in\{0,1,\ldots,k-1\}$ and $s-n+k+2\leq h$.
Let $\{ \mathbf{a_i}=( a_{i,1},a_{i,2},\ldots,a_{i,k} )^T \}_{i=1}^{s}$ be
the $sk$ data symbols in the first $s$ columns and $( a_{i,1},a_{i,2},\ldots,a_{i,k},\mathbf{P}_1^T\mathbf{a_i},$ $\ldots, \mathbf{P}_r^T\mathbf{a_i})^T$
be codeword $i$ of the $(n,k)$ MDS codes, where $i=1,2,\ldots,s$
and $\mathbf{P}_j^T=(\eta^{j-1},\eta^{2(j-1)},\ldots,\eta^{k(j-1)})$
with $j=1,2,\ldots,r,r=n-k$ and $\eta$ is a primitive element of $\mathbb{F}_q$.
Let $\{ \mathbf{b}=( b_{1},b_{2},\ldots,b_{k'} )^T \}$ be the $k'=k-h$
data symbols in the last column and
$( b_{1},b_{2},\ldots,b_{k'},\mathbf{Q}_1^T\mathbf{b},\ldots, \mathbf{Q}_{h+r}^T\mathbf{b})^T$
be a codeword of an $(n,k')$ MDS code, where
$\mathbf{Q}_j^T=(\eta^{j-1},\eta^{2(j-1)},\ldots,\eta^{k'(j-1)})$ with $j=1,2,\ldots,h+r$. Note that the total number of data symbols in this code is $sk+k'$.
\item Second, we add the {\em piggyback functions} of the symbols in the
first $s$ columns to the parity symbols in the last column, in order to reduce the
repair bandwidth. We divide the piggyback functions into two types: $(i)$ piggyback
functions of the symbols in the first $k'+1$ rows in the first $s$ columns; $(ii)$
piggyback functions of the symbols in the last $r+h-1$ rows in the first $s$ columns.
Fig. \ref{fig.1} shows the structure of two piggyback functions. For the first type of the piggyback functions, we add symbol $a_{i,j}$ (the symbol
in row $j$ and column $i$) to the parity symbol
$\mathbf{Q}_{2+(((j-1)s+i-1)\bmod(h+r-1))}^T\mathbf{b}$
(the symbol in row $k-h+2+(((j-1)s+i-1)\bmod(h+r-1))$ in the last column),
where $i\in\{1,2,\ldots,s\}$ and $j\in\{1,2,\ldots,k-h+1\}$.
For the second type of the piggyback functions, we add the symbol in row $j$
and column $i$ with $i\in\{1,2,\ldots,s\}$ and $j\in\{k-h+2,\ldots,k+r\}$
to the parity symbol $\mathbf{Q}_{t_{i,j}}^T\mathbf{b}$
(the symbol in row $k-h+t_{i,j}$ in the last column),
where
\begin{equation}
t_{i,j}=\left\{\begin{matrix}
i+j-k+h, \text{ if }\ i+j\leq n\\
i+j-n+1, \text{ if }\ i+j>n
\end{matrix}\right..
\label{eq:tij1}
\end{equation}
\end{enumerate}

The first piggybacking design described above is denoted by $\mathcal{C}(n,k,s,k')$.
When $h=0$, we have $k=k'$ and the created $s+1$ instances are codewords of $(n,k)$ MDS codes.
We will show the repair bandwidth in Section \ref{sec:3}.

We present the second piggybacking design as follows.
We create $s$ instances (in the first $s$ columns) of $(n,k)$ MDS codes over
finite field $\mathbb{F}_q$ and one additional empty column of length $n$, i.e.,
there is no data symbol in the last column, all the $n=k+r$ entries in the last columns
are piggyback functions. We design the $k+r$ piggyback functions in the last column
as follows. For $i\in\{1,2,\ldots,s\}$ and $j\in\{1,2,\ldots,k+r\}$, we add the
symbol in row $j$ and column $i$ to the symbol in row $\hat{t}_{i,j}$ in the last
column, where
\begin{equation}
\hat{t}_{i,j}=\left\{\begin{matrix}
i+j, \text{ if }\ i+j\leq n\\
i+j-n, \text{ if }\ i+j>n
\end{matrix}\right..
		\label{eq:tij2}
\end{equation}
We denote the second piggybacking design by $\mathcal{C}(n,k,s,k'=0)$,
the last parameter $k'=0$ denotes that there is not data symbol in the last column.
We will discuss the repair bandwidth in Section~\ref{sec:4}.

\begin{figure}[htpb]
	\centering
	\includegraphics[width=0.70\linewidth]{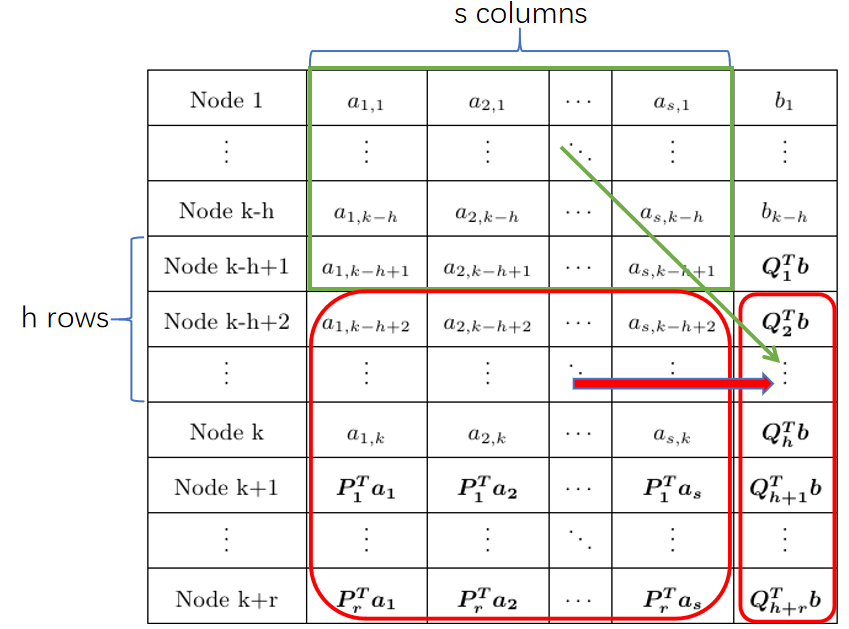}
	\caption{The structure of the first piggybacking design $\mathcal{C}(n,k,s,k')$, where $k'>0$.}
	\label{fig.1}
\end{figure}

Recall that in our first piggybacking design, the number of symbols to be added with piggyback functions in the last column
is $h-1+r$ and $h\geq s-r+2$ such that we can see that any two symbols used in computing both types of piggyback functions are from different nodes. Since
\begin{align*}
k-h+t_{i,j}=\left\{\begin{matrix}
k-h+i+j-k+h=i+j>j, \text{ when }\ i+j\leq n\\
k-h+i+j-n+1<j, \text{ when }\ i+j>n
\end{matrix}\right.,
\end{align*}
the symbol in row $j$ with $j\in\{k-h+2,k-h+3,\ldots,k+r\}$ and column $i$ with
$i\in\{1,2,\ldots,s\}$ is not added to the symbol in row $j$ and column $s+1$
in computing the second type of piggyback functions. In our second piggybacking design, since
\begin{align*}
\hat{t}_{i,j}=\left\{\begin{matrix} i+j>j, \text{ when }\ i+j\leq n\\
i+j-n<j, \text{ when }\ i+j>n
\end{matrix}\right.,
\end{align*}
the symbol in row $j$ with $j\in\{1,2,\ldots,k+r\}$ and column $i$ with
$i\in\{1,2,\ldots,s\}$ is not added to the symbol in row $j$ and column $s+1$
in computing the piggyback functions.

It is easy to see the MDS property of the first piggybacking design $\mathcal{C}(n,k,s,k')$.
We can retrieve all the other symbols in the first $s$ columns from any $k$ nodes (rows).
By computing all the piggyback functions and subtracting
all the piggyback functions from the corresponding parity symbols, we can retrieve
all the symbols in the last column.
Fig.~\ref{fig.2} shows an example of $\mathcal{C} (8,6,1,3)$.

\begin{figure}
	\centering
	\includegraphics[width=0.5\linewidth]{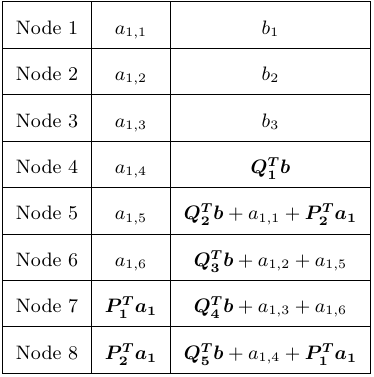}
	\caption{An example of $\mathcal{C} (n,k,s,k')$, where $(n,k,s,k')=(8,6,1,3)$.}
	\label{fig.2}
\end{figure}

Note that the piggyback function of the second piggybacking design is different from
that of the first piggybacking design. In the following of the section, we present
the repair method for the first piggybacking design. We will show the repair
method for the second piggybacking design in Section \ref{sec:4}.

For $i\in\{2,3,\ldots,h+r\}$, let $p_{i-1}$ be the piggyback function added on
the parity symbol $\mathbf{Q}_{i}^T\mathbf{b}$ and
$n_{i-1}$ be the number of symbols in the sum in computing piggyback function $p_{i-1}$.
{According to the design of piggyback functions, we have two set of symbols that
are used in computing the $h+r-1$ piggyback functions. The first set contains
$s(k-h+1)$ symbols (in the first $k-h+1$ rows and in the first $s$ columns) and
the second set contains $s(h+r-1)$ symbols (in the last $h+r-1$ rows and in the first
$s$ columns).
We have that the total number of symbols used in computing the $h+r-1$ piggyback functions
is $s(k+r)$, i.e.,
\begin{eqnarray}
	&&\sum_{i=1}^{h+r-1}n_i=s(k+r).\label{eq1}
\end{eqnarray}
In our first piggybacking design, the number of symbols used in computing each piggyback function is given
in the next lemma.
\begin{lemma}
In the first piggybacking design $\mathcal{C}(n,k,s,k')$ with $k'>0$, the number of symbols used in computing the piggyback function $p_{\tau}$ is
\begin{eqnarray}
	&&n_\tau= s+\left \lceil \frac{s(k-h+1)}{h+r-1} \right \rceil, \forall 1\leq \tau\leq (k-h+1)s-\lfloor \frac{(k-h+1)s}{h+r-1}\rfloor (h+r-1)\nonumber\\
	&&n_\tau= s+\left \lfloor \frac{s(k-h+1)}{h+r-1} \right \rfloor, \forall (k-h+1)s-\lfloor \frac{(k-h+1)s}{h+r-1}\rfloor (h+r-1)< \tau< h+r.\label{eq2}
\end{eqnarray}
\end{lemma}
\begin{proof}
In the design of the piggyback function, we add the symbol
in row $j$ and column $i$ for $i\in\{1,2,\ldots,s\}$ and $j\in\{1,2,\ldots,k-h+1\}$
($(k-h+1)s$ symbol in the first set) to the symbol in row  $k-h+2+(((j-1)s+i-1)\bmod(h+r-1))$
(piggyback function $p_{1+(((j-1)s+i-1)\bmod(h+r-1))}$) in the last column.
Therefore, we can see that the symbols in row $j$ and column $i$ are added to
$p_1$ for all $i\in\{1,2,\ldots,s\}$, $j\in\{1,2,\ldots,k-h+1\}$ and
$(j-1)s+i-1$ is a multiple of $h+r-1$.
Note that
\[
\{(j-1)s+i-1|i=1,2,\ldots,s,j=1,2,\ldots,k-h+1\}=\{0,1,\ldots,(k-h+1)s-1\},
\]
we need to choose the symbol in row $j$ and column $i$ for $i\in\{1,2,\ldots,s\}$,
$j\in\{1,2,\ldots,k-h+1\}$ such that $\eta$ is a multiple of $h+r-1$ for all
$\eta\in\{0,1,\ldots,(k-h+1)s-1\}$.
The number of symbols in the first set
used in computing $p_1$ is $\lceil \frac{(k-h+1)s}{h+r-1}\rceil$.
Given integer $\tau$ with $1\leq \tau\leq h+r-1$,
we add the symbol in row $j$ and column $i$ for $i\in\{1,2,\ldots,s\}$,
$j\in\{1,2,\ldots,k-h+1\}$ such that $\eta-\tau+1$ is a multiple of $h+r-1$ for all
$\eta\in\{0,1,\ldots,(k-h+1)s-1\}$ to $p_{\tau}$.
The number of symbols in the first set
used in computing $p_{\tau}$ is $\lceil \frac{(k-h+1)s}{h+r-1}\rceil$ if
$(k-h+1)s-\tau\geq \lfloor \frac{(k-h+1)s}{h+r-1}\rfloor (h+r-1)$ and
$\lfloor \frac{(k-h+1)s}{h+r-1}\rfloor$ if
$(k-h+1)s-\tau\leq \lfloor \frac{(k-h+1)s}{h+r-1}\rfloor (h+r-1)$.
Therefore, the number of symbols in the first set
used in computing $p_{\tau}$ is $\lceil \frac{(k-h+1)s}{h+r-1}\rceil$ if
$1\leq \tau\leq (k-h+1)s-\lfloor \frac{(k-h+1)s}{h+r-1}\rfloor (h+r-1)$ and
$\lfloor \frac{(k-h+1)s}{h+r-1}\rfloor$ if
$h+r-1\geq \tau\geq (k-h+1)s-\lfloor \frac{(k-h+1)s}{h+r-1}\rfloor (h+r-1)+1$.

For the $(h+r-1)s$ symbols in the second set, we add the symbol in row $j$
and column $i$ with $i\in\{1,2,\ldots,s\}$ and $j\in\{k-h+2,\ldots,k+r\}$
to the symbol in row $k-h+t_{i,j}$ (piggyback function {$p_{t_{i,j}-1}$})
in the last column, where $t_{i,j}$ is given in Eq. \eqref{eq:tij1}.
Consider the piggyback function $p_1$, i.e., $t_{i,j}=2$.
When $i=1$, according to Eq. \eqref{eq:tij1}, only when $j=k+r$
for $j\in\{k-h+2,\ldots,k+r\}$, we can obtain $t_{i,j}=2$.
When $i=2$, according to Eq. \eqref{eq:tij1}, only when $j=k+r-1$
for $j\in\{k-h+2,\ldots,k+r\}$, we can obtain $t_{i,j}=2$.
Similarly, for any $i$ with $i\in\{1,2,\ldots,s\}$, only when $j=k+r+1-i$
for $j\in\{k-h+2,\ldots,k+r\}$, we can obtain $t_{i,j}=2$. Since
$h\geq s-r+2$, we have $j=k+r+1-i\geq k+r+1-s>k-h+2$, which is within
$\{k-h+2,\ldots,k+r\}$. In other words, for any $i$ with
$i\in\{1,2,\ldots,s\}$, we can find one and only one $j$ with
$j\in\{k-h+2,\ldots,k+r\}$ such that $t_{i,j}=2$.
The number of symbols in the second set used in computing $p_{1}$ is $s$.
Similarly, we can show that the number of symbols in the second set
used in computing $p_{\tau}$ is $s$ for all $\tau=1,2,\ldots,h+r-1$.
Therefore, the total number of symbols
used in computing $p_{\tau}$ is $n_{\tau}=s+\lceil \frac{(k-h+1)s}{h+r-1}\rceil$ for
$\tau=1,2,\ldots,(k-h+1)s-\lfloor \frac{(k-h+1)s}{h+r-1}\rfloor (h+r-1)$
and $n_{\tau}=s+\lfloor \frac{(k-h+1)s}{h+r-1}\rfloor$ for
$\tau=(k-h+1)s-\lfloor \frac{(k-h+1)s}{h+r-1}\rfloor (h+r-1)+1,(k-h+1)s-\lfloor \frac{(k-h+1)s}{h+r-1}\rfloor (h+r-1)+2,\ldots,h+r-1$.
\end{proof}

Next lemma shows that any two symbols in one row in the first $s$ columns are used
in computing two different piggyback functions.
\begin{lemma}
In the first piggybacking design, if $s+2\leq h+r$, then the symbol in row $j$ in column $i_1$ and the symbol in row $j$ in column $i_2$
are used in computing two different piggyback functions, for any $j\in\{1,2,\ldots,k+r\}$
and $i_1\neq i_2\in\{1,2,\ldots,s\}$.
\label{lm:dif-piggy}
\end{lemma}
\begin{proof}
When $j\in\{1,2,\ldots,k-h+1\}$,
we add the symbol in row $j$ and column $i_1$ to the symbol
in row $k-h+2+(((j-1)s+i_1-1)\bmod(h+r-1))$ in the last column.
Similarly, the symbol in row $j$ and column $i_2$ is added to the symbol
in row $k-h+2+(((j-1)s+i_2-1)\bmod(h+r-1))$ in the last column.
Suppose that the two symbols in row $j$ and columns $i_1,i_2$ are added to the same
piggyback function, we obtain that $((j-1)s+i_1-1)\bmod(h+r-1)=((j-1)s+i_2-1)\bmod(h+r-1)$,
i.e., $i_1=i_2\bmod(h+r-1)$, which contradicts to $i_1\neq i_2\in\{1,2,\ldots,s\}$
and $s+2\leq h+r$.

When $j\in\{k-h+2,k-h+2,\ldots,k+r\}$, we add two symbols in row $j$
column $i_1$ and row $j$ column $i_2$ to the symbol in the last column
in row $k-h+t_{i_1,j}$ and row $k-h+t_{i_2,j}$, respectively,
where $i_1\neq i_2\in\{1,2,\ldots,s\}$ and
$$t_{i,j}=\left\{\begin{matrix}
i+j-k+h, \text{ if }\ i+j\leq n\\
i+j-n+1, \text{ if }\ i+j>n
\end{matrix}\right..$$
Suppose that the two symbols in row $j$ and columns $i_1,i_2$ are added to the same
piggyback function, we obtain that $t_{i_1,j}=t_{i_2,j}$.
If $i_1+j\leq n$ and $i_2+j\leq n$, we have that $i_1=i_2$ which contradicts to
$i_1\neq i_2$. If $i_1+j\leq n$ and $i_2+j> n$, we have that $i_1+r+h-1=i_2$ which contradicts to
$i_1\neq i_2\in\{1,2,\ldots,s\}$ and $s+2\leq h+r$. Similarly, we can obtain a
contradiction if $i_1+j> n$ and $i_2+j\leq n$.
If $i_1+j> n$ and $i_2+j> n$, we have that $i_1=i_2$ which contradicts to
$i_1\neq i_2$. Therefore, in our first piggybacking design, any two symbols in the same row are not used in computing
the same piggyback function.
\end{proof}

\subsection{Repair Process}
\label{sec:2.2}
In the fisrt piggybacking design, suppose that node $f$ fails, we present the repair procedure of node $f$ as
follows, where $f\in\{1,2,\ldots,k+r\}$.

We first consider that $f\in\{1,2,\ldots,k-h+1\}$, each of the first $s$ symbols
$\{ a_{1,f},a_{2,f},\ldots,a_{s,f} \}$ stored in node $f$ is used in
computing one piggyback function and we denote the corresponding piggyback
function associated with symbol $a_{i,f}$ by $p_{t_{i,f}}$, where $i=1,2,\ldots,s$
and $t_{i,f}\in\{1,2,\ldots,h+r-1\}$.
We download $k-h$ symbols in the last column from nodes
$\{1,2,\ldots,k-h+1\}\setminus\{f\}$ to recover $s+1$ symbols
$b_{f},\mathbf{Q}_{t_{1,f}+1}^T\mathbf{b},\mathbf{Q}_{t_{2,f}+1}^T\mathbf{b},
\ldots,\mathbf{Q}_{t_{s,f}+1}^T\mathbf{b}$ when $f\in\{1,2,\ldots,k-h\}$, or 
$\mathbf{Q}_{1}^T\mathbf{b},\mathbf{Q}_{t_{1,f}+1}^T\mathbf{b},\mathbf{Q}_{t_{2,f}+1}^T\mathbf{b},
\ldots,\mathbf{Q}_{t_{s,f}+1}^T\mathbf{b}$ when $f=k-h+1$, according to the MDS property of the last instance.
{By Lemma \ref{lm:dif-piggy}, any two
symbols in one row are used in computing two different piggyback functions.
The piggyback function $p_{t_{i,f}}$ is computed by $n_{t_{i,f}}$ symbols, where
one symbol is $a_{i,f}$ and the other $n_{t_{i,f}}-1$ symbols are not in node $f$ (row $f$).}
Therefore, we can repair the symbol $a_{i,f}$ in node $f$ by downloading the parity
symbol $\mathbf{Q}_{t_{i,f}+1}^T\mathbf{b}+p_{t_{i,f}}$ and $n_{t_{i,f}}-1$ symbols which are
used to compute $p_{t_{i,f}}$ except $a_{i,f}$, where $i=1,2,\ldots,s$.
The repair bandwidth is $k-h+\sum_{i=1}^{s}n_{t_{i,f}}$ symbols.

When $f\in\{k-h+2,k-h+3,\ldots,n\}$, each of the first $s$ symbols stored in
node $f$ is used in computing one piggyback function and we denote the corresponding
piggyback function associated with the symbol in row $f$ and column $i$ by
$p_{t_{i,f}}$, where $i\in\{1,2,\ldots,s\}$ and $t_{i,f}\in\{1,2,\ldots,h+r-1\}$.
We download $k-h$ symbols in the last column from nodes $\{1,2,\ldots,k-h\}$
to recover $s+1$ symbols $\mathbf{Q}_{f-k+h}^T\mathbf{b},
\mathbf{Q}_{t_{1,f}+1}^T\mathbf{b},\mathbf{Q}_{t_{2,f}+1}^T\mathbf{b},\ldots,
\mathbf{Q}_{t_{s,f}+1}^T\mathbf{b}$,
according to the MDS property of the last instance.
{Recall that any symbol in row $f$ in the first $s$ columns
is not used in computing the piggyback function in row $f$.}
We can recover the last symbol $\mathbf{Q}_{f-k+h}^T\mathbf{b}+p_{f-k+h-1}$ stored
in node $f$ by downloading $n_{f-k+h-1}$ symbols which are used to compute the
piggyback function $p_{f-k+h-1}$.
{Recall  that any two symbols in one row are used in computing
two different piggyback functions by Lemma \ref{lm:dif-piggy}.
The piggyback function $p_{t_{i,f}}$ is computed by $n_{t_{i,f}}$ symbols, where
one symbol is in row $f$ column $i$ and the other $n_{t_{i,f}}-1$ symbols are not
in node $f$ (row $f$).}
We can repair the symbol in row $f$ and column $i$, for $i\in\{1,2,\ldots,s\}$,
by downloading symbol $\mathbf{Q}_{t_{i,f}+1}^T\mathbf{b}+p_{t_{i,f}}$ and $n_{t_{i,f}}-1$
symbols which are used to compute $p_{t_{i,f}}$ except the symbol in row $f$
and column $i$.
The repair bandwidth is $k-h+n_{f-k+h-1}+\sum_{i=1}^{s}n_{t_{i,f}}$ symbols.

Consider the repair method of the code $\mathcal{C} (8,6,1,3)$ in Fig. \ref{fig.2}.
Suppose that node 1 fails, we can first download 3 symbols
$b_2,b_3,\mathbf{Q}_{1}^T\mathbf{b}$ to obtain the two symbols
$b_1,\mathbf{Q}_{2}^T\mathbf{b}$, according to the MDS property. Then, we
download the following 2 symbols
\[
\mathbf{Q}_{2}^T\mathbf{b}+a_{1,1}+\mathbf{P}_{2}^T\mathbf{a}_1,\mathbf{P}_{2}^T\mathbf{a}_1
\]
to recover $a_{1,1}$. The repair bandwidth of node 1 is 5 symbols.
Similarly, we can show that the repair bandwidth of any single-node failure
among nodes 2 to 4 is 5 symbols.

Suppose that node 5 fails, we can download the 3 symbols
$b_1,b_2,b_{3}$ to obtain $\mathbf{Q}_{2}^T\mathbf{b},\mathbf{Q}_{3}^T\mathbf{b}$,
according to the MDS poverty. Then, we download the 2 symbols
$a_{1,1},\mathbf{P}_{2}^T\mathbf{a}_1$ to recover $\mathbf{Q}_{2}^T\mathbf{b}+p_1$.
Finally, we download the 2 symbols
$\mathbf{Q}_{3}^T\mathbf{b}+p_2,a_{1,2}$
to recover $a_{1,5}$.
The repair bandwidth of node 5 is 7 symbols.
Similarly, we can show that the repair bandwidth of any single-node failure
among nodes 6 to 8 is 7 symbols.

\subsection{Average Repair Bandwidth Ratio of Code $\mathcal{C} (n,k,s,k'), k'>0$}
\label{sec:2.3}
Define the {\em average repair bandwidth} of data nodes (or parity nodes or all nodes) as the ratio of the
summation of repair bandwidth for each of $k$ data nodes (or $r$ parity nodes or all $n$ nodes) to the number of
data nodes $k$ (or the number of parity nodes $r$ or the number of all nodes $n$).
Define the {\em average repair bandwidth ratio} of data nodes (or parity nodes or all nodes) as the ratio of the
average repair bandwidth of $k$ data nodes (or $r$ parity nodes or all $n$ nodes) to the number of
data symbols.

In the following, we present an upper bound of the average repair bandwidth ratio
of all $n$ nodes, denoted by $\gamma^{all}$, for the proposed codes $\mathcal{C} (n,k,s,k')$ when $k'>0$.

\begin{theorem}
	\label{th1}
When $k'>0$, the average repair bandwidth ratio of all $n$ nodes, $\gamma^{all}$, of
codes $\mathcal{C} (n,k,s,k')$, is upper bounded by
	\begin{eqnarray}
		\gamma^{all}&\leq&\frac{(u+s)^2(h+r-1)}{(k+r)(sk+k-h)}+\frac{k-h+s}{sk+k-h},\nonumber
	\end{eqnarray}
	where $u=\left \lceil \frac{s(k-h+1)}{h+r-1} \right \rceil$.
\end{theorem}
\begin{proof}
{Suppose that node $f$ fails, where $f\in\{1,2,\ldots,n\}$, we
will count the repair bandwidth of node $f$ as follows.
Recall that the symbol in row $f$ and column $i$ is used to
compute the piggyback function $p_{t_{i,f}}$, where $f\in\{1,2,\ldots,n\}$
and $i\in\{1,2,\ldots,s\}$.} Recall also that the number of symbols in the sum in
computing piggyback function $p_{t_{i,f}}$ is $n_{t_{i,f}}$.
When $f\in\{1,2,\ldots,k-h+1\}$, according to the repair method in Section \ref{sec:2.2}, the repair bandwidth
of node $f$ is $(k-h+\sum_{i=1}^{s}n_{t_{i,f}})$ symbols.
When $f\in\{k-h+2,\ldots,n\}$, according to the repair method in Section \ref{sec:2.2},
the repair bandwidth of node $f$ is $(k-h+n_{f-k+h-1}+\sum_{i=1}^{s}n_{t_{i,f}})$ symbols.
The summation of the repair bandwidth for each of the $n$ nodes is
	 \begin{eqnarray}
&&\sum_{f=1}^{k-h+1}(k-h+\sum_{i=1}^{s}n_{t_{i,f}})+
\sum_{f=k-h+2}^{k+r}(k-h+n_{f-k+h-1}+\sum_{i=1}^{s}n_{t_{i,f}})\nonumber\\
=&&(k+r)(k-h)+\sum_{f=1}^{k+r}(\sum_{i=1}^{s}n_{t_{i,f}})+\sum_{f=k-h+2}^{k+r}n_{f-k+h-1}.\label{eq:rep-sum}
%=&&(k+r)(k-h-1)+\sum_{i=1}^{h+r}n_i^2+\sum_{f=k-h+1}^{k+r}n_{f-k+h}.
	 \end{eqnarray}
Next, we show that
\begin{equation}
\sum_{f=1}^{k+r}(\sum_{i=1}^{s}n_{t_{i,f}})=\sum_{i=1}^{h+r-1}n_i^2.
\label{eq:rep-sum1}
\end{equation}
Note that $\sum_{i=1}^{k+r}(\sum_{i=1}^{s}n_{t_{i,f}})$ is the summation of the
repair bandwidth for each of the $(k+r)s$ symbols in the first $s$ columns.
The $(k+r)s$ symbols are used to compute the $h+r-1$ piggyback functions and each symbol
is used for only one piggyback function. For $i=1,2,\ldots,h+r-1$, the piggyback function
$p_i$ is the summation of the $n_i$ symbols in the first $s$ columns and can recover
any one of the $n_i$ symbols (used in computing $p_i$) with repair bandwidth
$n_i$ symbols. Therefore, the summation of the repair bandwidth for each of the $n_i$
symbols (used in computing $p_i$) is $n_i^2$. In other words, the summation of the
repair bandwidth for each of the $(k+r)s$ symbols in the first $s$ columns
is the summation of the repair bandwidth for each of all the $(k+r)s$ symbols used for
computing all $h+r-1$ piggyback functions, i.e., Eq. \eqref{eq:rep-sum1} holds.

By Eq. \eqref{eq1}, we have  $\sum_{f=k-h+2}^{n}n_{f-k+h-1}=\sum_{i=1}^{h+r-1}n_i=s(k+r)$.
By Eq. \eqref{eq2}, we have $n_i\leq u+s, \forall i\in\{1,2,\ldots,h+r\}$,
where $u=\left \lceil \frac{s(k-h+1)}{h+r-1} \right \rceil$. According to Eq. \eqref{eq:rep-sum}
and Eq. \eqref{eq:rep-sum1}, we have
 	 \begin{eqnarray}
 		\gamma^{all}&=&\frac{(k+r)(k-h)+\sum_{i=1}^{h+r-1}n_i^2}{(k+r)(sk+k-h)}
 		+\frac{\sum_{f=k-h+2}^{n}n_{f-k+h-1}}{(k+r)(sk+k-h)}\nonumber\\
 		&=&\frac{(k+r)(k-h+s)+\sum_{i=1}^{h+r-1}n_i^2}{(k+r)(sk+k-h)}\nonumber\\
 &\leq&\frac{k-h+s}{sk+k-h}+\frac{(u+s)^2(h+r-1)}{(k+r)(sk+k-h)}.\nonumber
 	\end{eqnarray}
\end{proof}

Define {\em storage overhead} to be the ratio of total number of symbols stored in the $n$
nodes to the total number of data symbols. We have that the storage overhead $s^*$ of
codes $\mathcal{C}(n,k,s,k')$ satisfies that
 	 \begin{eqnarray}
 		&&\frac{k+r}{k}\leq s^*=\frac{(s+1)(k+r)}{sk+k-h}\leq\frac{(s+1)(k+r)}{sk}=(\frac{s+1}{s})\cdot\frac{k+r}{k}.\nonumber
 	\end{eqnarray}

\section{Piggybacking Codes $\mathcal{C}(n,k,s,k'=k)$}
\label{sec:3}
In this section, we consider the special case of codes $\mathcal{C}(n,k,s,k')$ with
$k'=k$. When $k'=k$, we have $s\leq r-2$ and the created $s+1$ instances are
codewords of $(n,k)$ MDS codes and the codes $\mathcal{C}(n,k,s,k'=k)$ are MDS codes.
The structure of $\mathcal{C}(n,k,s,k'=k)$ is shown in Fig. \ref{fig.3}.
\begin{figure}[htpb]
	\centering
	\includegraphics[width=0.60\linewidth]{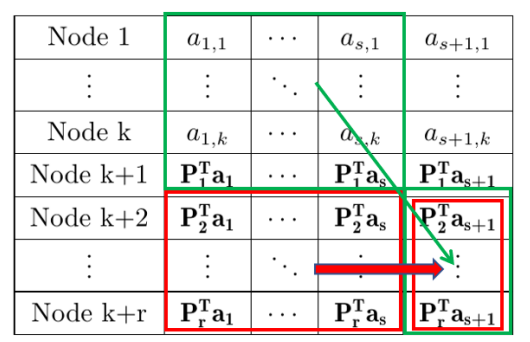}
	\caption{The design of code $\mathcal{C}(n,k,s,k'=k),s\leq r-2$.}
	\label{fig.3}
\end{figure}

In $\mathcal{C}(n,k,s,k'=k)$, we have $r-1$ piggyback functions $\{p_i\}_{i=1}^{r-1}$,
and each piggyback function $p_i$ is a linear combination of $n_i$ symbols that are
located in the first $s$ columns of the $n\times (s+1)$ array, where
$i\in\{1,2,\ldots,r-1\}$. According to Eq. \eqref{eq1}, we have
\begin{eqnarray}
	&&\sum_{i=1}^{r-1}n_i=s(k+r).\label{eq7}
\end{eqnarray}

The average repair bandwidth ratio of all nodes of $\mathcal{C}(n,k,s,k'=k)$ is
given in the next theorem.

\begin{theorem}
	\label{th2}
The lower bound and the upper bound of the average repair bandwidth ratio
of all nodes $\gamma^{all}_{0}$ of $\mathcal{C}(n,k,s,k'=k)$ is
	\begin{eqnarray}
		&&\gamma^{all}_{0,min}=\frac{k+s}{(s+1)k}+\frac{s^2(k+r)}{(r-1)(s+1)k} \text{ and }\label{eq8}\\
		&&\gamma^{all}_{0,max}=\gamma^{all}_{0,min}+\frac{r-1}{4k(k+r)(s+1)},\label{eq9}
	\end{eqnarray}
respectively.
\end{theorem}
\begin{proof}
By Eq. \eqref{eq:rep-sum}, the summation of the repair bandwidth for each of the $n$ nodes is
	\begin{eqnarray}
		&&(k+r)k+\sum_{i=1}^{r-1}n_i^2+\sum_{i=1}^{r-1}n_{i}.\nonumber
	\end{eqnarray}
By Eq. \eqref{eq7}, we have
	\begin{eqnarray}
		\gamma^{all}_0&=&\frac{(k+r)k+\sum_{i=1}^{r-1}n_i^2+\sum_{i=1}^{r-1}n_{i}}{(k+r)(s+1)k}\nonumber\\
		&=&\frac{(k+r)(k+s)+\sum_{i=1}^{r-1}n_i^2}{(k+r)(s+1)k}\nonumber\\
		&=&\frac{(k+r)(k+s)+\frac{(\sum_{i=1}^{r-1}n_{i})^2+\sum_{i\neq j}(n_i-n_j)^2}{r-1}}{(k+r)(s+1)k}.\nonumber
	\end{eqnarray}
	Note that $\sum_{i\neq j}(n_i-n_j)^2=t(r-1-t)$ by Eq. \eqref{eq2}, where $t=s(k+1)-\left \lfloor \frac{s(k+1)}{r-1} \right \rfloor(r-1)$. According to Eq. \eqref{eq7}, we have
	\begin{eqnarray}
			\gamma^{all}_0&=&\frac{(k+r)(k+s)+\frac{(\sum_{i=1}^{r-1}n_{i})^2+t(r-1-t)}{r-1}}{(k+r)(s+1)k}\nonumber\\
			&=&\frac{k+s}{(s+1)k}+\frac{s^2(k+r)}{(r-1)(s+1)k}+\frac{t(r-1-t)}{(k+r)(s+1)(r-1)k}.\nonumber
	\end{eqnarray}
By the mean inequality, we have $0\leq t(r-1-t)\leq\frac{(r-1)^2}{4}$ and we can
further obtain that
	\begin{eqnarray}
		&&\frac{k+s}{(s+1)k}+\frac{s^2(k+r)}{(r-1)(s+1)k}\leq\gamma^{all}_{0}\nonumber\\ &&\leq\frac{k+s}{(s+1)k}+\frac{s^2(k+r)}{(r-1)(s+1)k}+\frac{r-1}{4k(k+r)(s+1)}.\nonumber
	\end{eqnarray}
\end{proof}

According to Theorem \ref{th2}, the difference between the lower bound and the upper bound
of the average repair bandwidth ratio satisfies that
	\begin{eqnarray} |\gamma^{all}_{0}-\gamma^{all}_{0,min}|&\leq&|\gamma^{all}_{0,min}-\gamma^{all}_{0,max}|=
\frac{r-1}{4k(k+r)(s+1)}\leq\frac{r-1}{8k(k+r)}.\label{eq10}
\end{eqnarray}
When $r\ll k$, the difference between $\gamma^{all}_{0}$ and $\gamma^{all}_{0,min}$
can be ignored. When $r\ll k$, we present the repair bandwidth for $\mathcal{C}(n,k,s,k'=k)$
as follows.

\begin{corollary}
	\label{col3}
Let $r\ll k$ and $k\rightarrow +\infty$, then the minimum value of the average repair bandwidth ratio $\gamma^{all}_{0}$
of $\mathcal{C}(n,k,s,k'=k)$ is achieved when $s=\sqrt{r}-1$.
\end{corollary}
\begin{proof}
When $r\ll k$ and $k\rightarrow +\infty$, we have that $\underset{k\rightarrow +\infty}{lim}\gamma_{0}^{all}
=\underset{k\rightarrow +\infty}{lim}\gamma_{0,min}^{all}$ by Eq. \eqref{eq10}
and by Eq. \eqref{eq8}, we can further obtain that
	\begin{eqnarray}
\underset{k\rightarrow +\infty}{lim}\gamma_{0}^{all}
		=\underset{k\rightarrow +\infty}{lim}\gamma_{0,min}^{all}
		=\frac{s^2}{(r-1)(s+1)}+\frac{1}{s+1}.\label{eq11}
	\end{eqnarray}
We can compute that
	\begin{eqnarray}
		&&\frac{\partial\underset{k\rightarrow +\infty}{lim}\gamma_{0}^{all}}{\partial s}=\frac{(s+1)^2-r}{(s+1)^2(r-1)}.\nonumber
	\end{eqnarray}
If $s>\sqrt{r}-1$, then $\frac{\partial\gamma(r,s)}{\partial s}>0$;
if $s<\sqrt{r}-1$, then $\frac{\partial\gamma(r,s)}{\partial s}<0$;
if $s=\sqrt{r}-1$, then $\frac{\partial\gamma(r,s)}{\partial s}=0$.
Therefore, when $s=\sqrt{r}-1$, $\underset{k\rightarrow +\infty}{lim}\gamma_{0}^{all}$
achieves the minimum value.
\end{proof}

Note that $s$ should be a positive integer, we can let $s=\left \lfloor \sqrt{r}-1 \right \rfloor$
or $s=\left \lceil \sqrt{r}-1 \right \rceil$, and then choose the minimum value
of the average repair bandwidth ratio $\gamma^{all}_{0}$
for $\mathcal{C}(n,k,s,k'=k)$.

\begin{figure}
	\centering
	\includegraphics[width=0.55\linewidth]{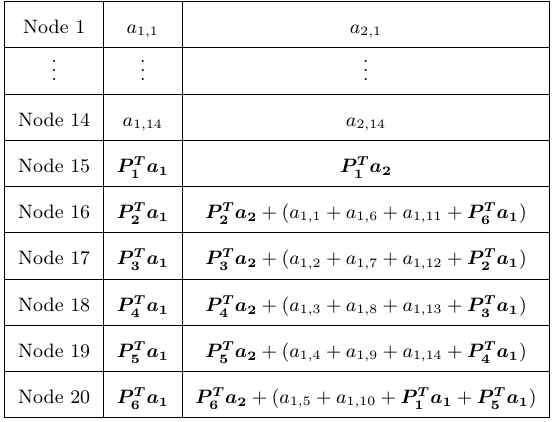}
	\caption{An example of code $\mathcal{C}(n,k,s,k'=k)$, where $(n,k,s)=(20,14,1)$.}
	\label{fig.4}
\end{figure}

Consider a specific example of code $\mathcal{C} (n=20,k=14,s=1,k'=14)$ in Fig. \ref{fig.4}.
Suppose that node 1 fails, we can first download 14 symbols
$a_{2,2},a_{2,3},\ldots,a_{2,14},\mathbf{P}_{1}^T\mathbf{a_2}$ to obtain the two symbols
$a_{2,1},\boldsymbol{P_2^Ta_2}$, according to the MDS property. Then, we
download the following 4 symbols
\[
\boldsymbol{P_2^Ta_2}+(a_{1,1}+a_{1,6}+a_{1,11}+\boldsymbol{P_6^Ta_1}),
a_{1,6},a_{1,11},\mathbf{P}_{6}^T\mathbf{a}_1
\]
to recover $a_{1,1}$. The repair bandwidth of node 1 is 18 symbols.
Similarly, we can show that the repair bandwidth of any single-node failure
among nodes 2 to 15 is 18 symbols.

Suppose that node 16 fails, we can download the 14 symbols
$a_{2,1},a_{2,2},\ldots,a_{2,14}$ to obtain $\mathbf{P}_{2}^T\mathbf{a_{2}},\mathbf{P}_{3}^T\mathbf{a_2}$,
according to the MDS poverty. Then, we download the 4 symbols
$\boldsymbol{P_3^Ta_2}+(a_{1,2}+a_{1,7}+a_{1,12}+\boldsymbol{P_2^Ta_1}),a_{1,2},a_{1,7},a_{1,12}$ to recover $\mathbf{P}_{2}^T\mathbf{a_1}$.
Finally, we download the 4 symbols
$a_{1,1},a_{1,6},a_{1,11},\boldsymbol{P_6^Ta_1}$
to recover $\mathbf{P}_{2}^T\mathbf{a}_2+p_1$.
The repair bandwidth of node 16 is 22 symbols.
Similarly, we can show that the repair bandwidth of any single-node failure
among nodes 17 to 20 is 22 symbols. Therefore, we know that in this example, the average repair bandwidth ratio of all nodes is $\frac{15*18+5*22}{20*28}\approx 0.68$.

\section{Piggybacking Codes $\mathcal{C}(n,k,s,k'=0)$}
\label{sec:4}
In this section, we consider the special case of codes $\mathcal{C}(n,k,s,k'=0)$
with $n\geq s+1$ based on the second piggybacking design. Recall that
there is no data symbol in the last column and we add the $n$ piggyback
functions in the last column. Here, for $i\in\{1,2,\dots,n\}$, we use $p_i$ to represent the piggyback function in the last column in row (node) $i$. Fig. \ref{fig.5} shows the structure of codes $\mathcal{C}(n,k,s,k'=0)$.

\begin{figure}[htpb]
	\centering
	\includegraphics[width=0.55\linewidth]{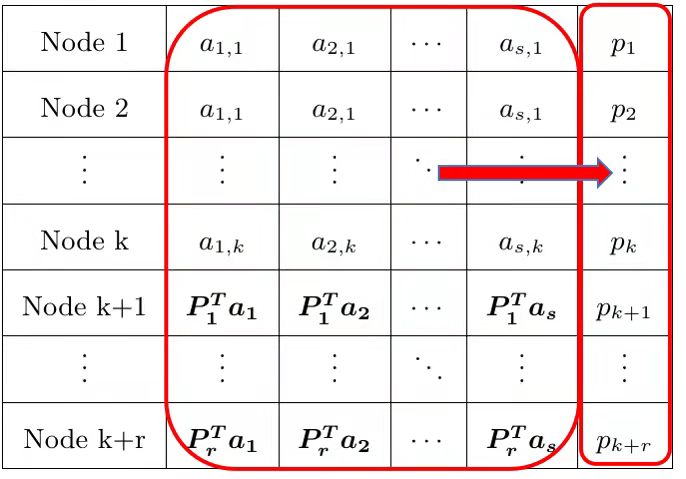}
	\caption{The structure of codes $\mathcal{C}(n,k,s,k'=0)$.}
	\label{fig.5}
\end{figure}

For notational convenience, we denote the parity symbol
$\mathbf{P}_j^T\mathbf{a}_i$ by $a_{i,k+j}$ in the following,
where $1\leq j\leq r, 1\leq i\leq s$.
Given an integer $x$ with $-s+1\leq x\leq k+r+s$, we define $\overline{x}$ by
\[
\overline{x}=\left\{\begin{matrix}
	x+k+r, \text{ if }\ -s+1\leq x\leq0\\
	x, \text{ if }\ 1\leq x\leq k+r\\
	x-k-r, \text{ if }\ k+r+1\leq x\leq k+r+s
\end{matrix}\right..
\]
According to the design of piggyback functions in Section \ref{sec:2.1},
{the symbol $a_{i,j}$ is used to compute the piggyback
function $p_{\overline{i+j}}$ for $i\in\{1,2,\ldots,s\}$ and $j\in\{1,2,\ldots,n\}$.
Therefore, we can obtain that the piggyback function $p_j=\sum_{i=1}^{s}a_{i,\overline{j-i}}$
for $1\leq j\leq k+r$. For any $j\in\{1,2,\ldots,k+r\}$ and $i_1\neq i_2\in\{1,2,\ldots,s\}$,
the two symbols in row $j$, columns $i_1$ and $i_2$ are added to two different
piggyback functions $p_{\overline{i_1+j}}$ and $p_{\overline{i_2+j}}$, respectively,
since $\overline{i_1+j}\neq \overline{i_2+j}$ for $n\geq s+1$.} The next theorem shows the repair bandwidth of codes $\mathcal{C}(n,k,s,k'=0)$.

\begin{theorem}
	\label{th4}
The repair bandwidth of codes $\mathcal{C}(n,k,s,k'=0)$ is $s+s^2$.
\end{theorem}
\begin{proof}
Suppose that node $f$ fails, where $f\in \{1,2,\ldots,n\}$. {Similar to}
the repair method in Section \ref{sec:2.2}, we can first repair the piggyback function
$p_f$ (the symbol in row $f$ and column $s+1$) by downloading $s$ symbols
$\{a_{j,\overline{f-j}}\}_{j=1}^s$ to repair the symbol $p_f$. Note that the symbol
$a_{j,f}$ is used to compute the piggyback function $p_{\overline{j+f}}$
for $j\in\{1,2,\ldots,s\}$, we can recover the symbol $a_{j,f}$ by downloading
$p_{\overline{j+f}}$ and the other $s-1$ symbols used in computing
$p_{\overline{j+f}}$ except $a_{j,f}$. Therefore, the repair bandwidth of node
$f$ is $s+s^2$ symbols.
\end{proof}

By Theorem \ref{th4}, the average repair bandwidth ratio of $\mathcal{C}(n,k,s,k'=0)$
is $\frac{s+1}{k}$. The storage overhead of  $\mathcal{C}(n,k,s,k'=0)$ is $$\frac{(s+1)(k+r)}{sk}=(1+\frac{1}{s})\cdot\frac{k+r}{k}.$$
We show in the next theorem that our codes $\mathcal{C}(n,k,s,k'=0)$ can
recover any $r+1$ failures under some condition.

\begin{theorem}
	\label{th5}
If $k>(s-1)(r+1)+1$, then the codes $\mathcal{C}(n,k,s,k'=0)$ can
recover any $r+1$ failures.
\end{theorem}
\begin{proof}
Suppose that $r+1$ nodes $f_1,f_2,\ldots,f_{r+1}$ fail, where
$1\leq f_1<f_2<\cdots<f_{r+1}\leq k+r$. In the following, we present a repair
method to recover the failed $r+1$ nodes.

For $i\in\{1,2,\ldots,r\}$, let $t_i$ be the number of surviving nodes between
two failed nodes $f_i$ and $f_{i+1}$, and $t_{r+1}$ be the number of surviving
nodes between nodes $f_{r+1}$ and $k+r$ plus the number of surviving nodes between nodes
1 and $f_1$, i.e., $t_i=f_{i+1}-f_i-1$ and
$t_{r+1}=k+r-f_{r+1}+f_1-1$. It is easy to see that
\begin{eqnarray}
&&\sum_{i=1}^{r+1}t_i=k-1.\nonumber
\end{eqnarray}
Let $t_{\max}=\max\{t_1,t_{2},\ldots,t_{r+1}\}$ and without loss of generality,
we assume that $t_{\max}=t_j$ with $j\in\{1,2,\ldots,r+1\}$. We have $t_j=t_{\max}\geq t_i$
for $i=1,2,\ldots,r+1$ and
\begin{eqnarray}
&&(r+1)t_j\geq\sum_{i=1}^{r+1}t_i=k-1>(s-1)(r+1),\nonumber
\end{eqnarray}
where the last inequality comes from that $k>(s-1)(r+1)+1$.
Therefore, we obtain that $t_j\geq s$.

We are now ready to describe the repair method for the failed $r+1$ nodes.
We can first repair one symbol stored in the failed node $f_j$, some symbols
that are used to repair the symbol in node $f_j$ are shown in Fig. \ref{fig.6}.
Recall that the symbol $p_{\overline{f_j+s}}$ is located column $s+1$
in node $\overline{f_j+s}$. We claim that node $\overline{f_j+s}$ is not
a failed node, since $t_j\geq s$. Recall also that $p_{\overline{f_j+s}}=a_{s,f_j}+\sum_{i=1}^{s-1}a_{i,\overline{f_j+s-i}}$.
We claim that node $\overline{f_j+s-i}$ is not a failed node, for all
$1\leq i\leq s-1$, since $t_{j}\geq s$.
Therefore, we can download the $s$ symbols $p_{\overline{f_j+s}},
\{a_{i,\overline{f_j+s-i}}\}_{i=1}^{s-1}$ to recover the symbol
$a_{s,f_j}$.
Once the symbol $a_{s,f_j}$ is recovered, we can recover the other
failed $r$ symbols $\{a_{s,f_i}\}_{i=1,2,\ldots,j-1,j+1,\ldots,r+1}$ in column $s$
in nodes $f_1,\ldots,f_{j-1},f_{j+1},\ldots,f_{r+1}$, according to the MDS property
of the MDS codes.

\begin{figure}[htpb]
	\centering
	\includegraphics[width=0.7\linewidth]{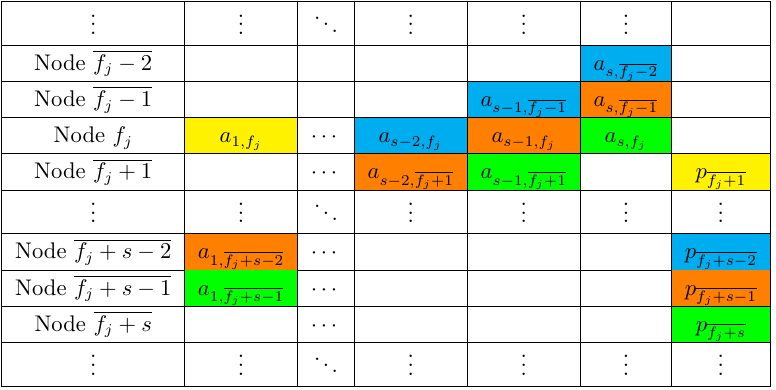}
	\caption{Piggyback functions of the codes $\mathcal{C}(n,k,s,k'=0)$,
where one piggyback function and the corresponding symbols that are used to compute
the piggyback function are with the same color.}
	\label{fig.6}
\end{figure}

Next, we present a repair method to recover the symbol $a_{s-1,f_j}$.
First, recall that
\[
p_{\overline{f_j+s-1}}=a_{s-1,f_j}+a_{s,\overline{f_j-1}}+\sum_{i=1}^{s-2}a_{i,\overline{f_j+s-1-i}},
\]
node $\overline{f_j+s-1-i}$ is not a failed node for all $1\leq i\leq s-2$
and the symbol $a_{s,\overline{f_j-1}}$ has already recovered. Therefore, we can
recover $a_{s-1,f_j}$ by downloading the following $s$ symbols
\[
p_{\overline{f_j+s-1}}, a_{s,\overline{f_j-1}}, \{a_{i,\overline{f_j+s-1-i}}\}_{i=1}^{s-2}.
\]
Once the symbol $a_{s-1,f_j}$ is recovered, we can recover the other
failed $r$ symbols $\{a_{s-1,f_i}\}_{i=1,2,\ldots,j-1,j+1,\ldots,r+1}$ in column $s-1$
in nodes $f_1,\ldots,f_{j-1},f_{j+1},\ldots,f_{r+1}$, according to the MDS property
of the MDS codes.

Similarly, we can recover $a_{s-\ell,f_j}$ by downloading the following $s$ symbols 
\[
p_{\overline{f_j+s-\ell}}, \{a_{s-l+i,\overline{f_j-i}}\}_{i=1}^{l}, \{a_{s-l-i,\overline{f_j+i}}\}_{i=1}^{s-l-1},
\]
and further recover the other
failed $r$ symbols $\{a_{s-\ell,f_i}\}_{i=1,2,\ldots,j-1,j+1,\ldots,r+1}$ in column $s-\ell$
in nodes $f_1,\ldots,f_{j-1},f_{j+1},\ldots,f_{r+1}$, according to the MDS property
of the MDS codes, where $\ell=1,2,\ldots,s-1$.

We have recovered $(r+1)s$ symbols stored in the $r+1$ failed nodes and can recover
all the other $r+1$ symbols in the $r+1$ failed nodes in the last column.
\end{proof}

Recall that an $(n,k,g)$ Azure-LRC code \cite{huang2012} first computes
$n-k-g$ global parity symbols
by encoding all the $k$ data symbols, then divides the $k$ data symbols
into $g$ groups each with
$\frac{k}{g}$ symbols (suppose that $k$ is a multiple of $g$) and compute one local parity symbol for each group.
Each of the obtained $n$ symbols is stored in one node and an $(n,k,g)$ Azure-LRC code can tolerant any $n-k-g+1$-node failures.
In the following theorem, we will show that the proposed  codes
$\mathcal{C}(n-g,k,s,k'=0)$ have lower repair bandwidth than the
$(n,k,g)$ Azure-LRC code \cite{huang2012} under the condition that both the fault-tolerance and
the storage overhead of two codes are the same.

\begin{theorem}
	\label{th6}
If $2g>n-k+1$ and $n^2-k^2<kg\cdot (n-k-g+1)$, then the proposed codes $\mathcal{C}(n-g,k,\frac{n-g}{g},k'=0)$ (suppose that $\frac{n-g}{g}$ is an integer)
have strictly less repair bandwidth
than that of $(n,k,g)$ Azure-LRC code, under the condition that both the fault-tolerance and
the storage overhead of two codes are the same.
\end{theorem}
\begin{proof}
Recall that the storage overhead of $(n,k,g)$ Azure-LRC code is $\frac{n}{k}$ and
the fault-tolerance of $(n,k,g)$ Azure-LRC code is $n-k-g+1$. We should determine the parameter
$s$ for $\mathcal{C}(n-g,k,s,k'=0)$ such that the storage overhead
is $\frac{n}{k}$ and the fault-tolerance is $n-k-g+1$.

When $s=\frac{n-g}{g}$, we have that the storage overhead of
code $\mathcal{C}(n-g,k,s,k'=0)$ is
\begin{eqnarray}
	\frac{(s+1)\cdot(k+r)}{sk}=\frac{(\frac{n-g}{g}+1)\cdot(n-g)}{(\frac{n-g}{g})\cdot k}=\frac{n}{k},\nonumber
\end{eqnarray}
which is equal to the storage overhead of $(n,k,g)$ Azure-LRC code. By assumption, we have that $2g>n-k+1$, then we can obtain that
\begin{align*}
(s-1)(r+1)+1=(\frac{n-g}{g}-1)(n-k-g+1)+1<(\frac{n-g}{g}-1) g+1=n-2g+1<k.
\end{align*}
Therefore, the condition in Theorem \ref{th5} is satisfied, the fault-tolerant of our
code $\mathcal{C}(n-g,k,\frac{n-g}{g},k'=0)$ is $r+1=n-k-g+1$ and the code length of $\mathcal{C}(n-g,k,\frac{n-g}{g},k'=0)$
is $n-g$.
In the following,
we show that the average repair bandwidth ratio of all nodes of our
$\mathcal{C}(n-g,k,\frac{n-g}{g},k'=0)$ is strictly less than that of $(n,k,g)$ Azure-LRC code.

Let the average repair bandwidth ratio of all nodes of $(n,k,g)$ Azure-LRC code
and $\mathcal{C}(n-g,k,\frac{n-g}{g},k'=0)$ be $\gamma_{1}$ and $\gamma_{2}$, respectively.
In $(n,k,g)$ Azure-LRC code, we can repair any symbol in a group by downloading the other
$\frac{k}{g}$ symbols in the group and repair any global parity symbol by downloading
the $k$ data symbols. The average repair bandwidth ratio of all nodes of $(n,k,g)$ Azure-LRC code is
	\begin{eqnarray}
	&&\gamma_{1}=\frac{(k+g)\cdot \frac{k}{g}+(n-k-g)\cdot k}{nk}=\frac{(n-k-g+1)\cdot g+k}{ng}.\nonumber
	\end{eqnarray}
According to Theorem \ref{th4}, the average repair bandwidth ratio of all nodes of
$\mathcal{C}(n-g,k,\frac{n-g}{g},k'=0)$ is $\frac{s+1}{k}$. We have that
	\begin{eqnarray}
	&&\gamma_{2}<\gamma_{1}\Leftrightarrow\frac{s+1}{k}<\frac{(n-k-g+1)\cdot g+k}{ng}\nonumber\\
	&&\Leftrightarrow\frac{\frac{n-g}{g}+1}{k}<\frac{(n-k-g+1)\cdot g+k}{ng}\nonumber\\
	&&\Leftrightarrow n^2-k^2<kg\cdot (n-k-g+1).\nonumber
	\end{eqnarray}
By the assumption, we have that $\gamma_{2}<\gamma_{1}$.
\end{proof}

By Theorem \ref{th6}, the storage overhead of $\mathcal{C}(n-g,k,\frac{n-g}{g},k'=0)$
and $(n,k,g)$ Azure-LRC code are the same and $\mathcal{C}(n-g,k,\frac{n-g}{g},k'=0)$
have strictly less repair bandwidth than that of $(n,k,g)$ Azure-LRC code, when the condition
is satisfied. Since the storage overhead and the repair bandwidth of $\mathcal{C}(n,k+g,\frac{n}{g},k'=0)$
are strictly less than that of $\mathcal{C}(n-g,k,\frac{n-g}{g},k'=0)$.
We thus obtain that the proposed codes $\mathcal{C}(n,k+g,\frac{n}{g},k'=0)$ have better
performance than that of $(n,k,g)$ Azure-LRC code in terms of both storage overhead
and repair bandwidth, when $2g>n-k+1$ and $n^2-k^2<kg\cdot (n-k-g+1)$.

Note that, with $\mathcal{C}(n,k+g,\frac{n}{g},k'=0)$, one can repair
any single-node failure by accessing $\frac{2n}{g}$ helper nodes; however,  one  only need to
access $\frac{k}{g}$ helper nodes in repairing any symbol in a group and access $k$ helper
nodes in repairing any global parity symbol with $(n,k,g)$ Azure-LRC code.

Recall that an $(n,k,g)$ optimal-LRC first encodes $k$ data symbols to obtain
$n-k-g$ global parity symbols, then divides the $n-g$ symbols (including $k$
data symbols and $n-k-g$ global parity symbols) into $g$ groups each with
$(n-g)/g$ symbols and encodes one local parity symbol for each group, where
$n-g$ is a multiple of $g$. Each of the obtained $n$ symbols is stored in one node and an $(n,k,g)$ optimal-LRC code can tolerant any $n-k-g+1$ node failures. The repair bandwidth of any single-node failure of
optimal-LRC is $(n-g)/g$ symbols. Next theorem shows that our codes
$\mathcal{C}(n,k+g,\frac{n-g}{g},k'=0)$ have better performance than 
$(n,k,g)$ optimal-LRC code, in terms of both storage overhead and repair bandwidth.

\begin{theorem}
	\label{th7}
	If $2g>n-k+1$, then the proposed codes $\mathcal{C}(n,k+g,\frac{n-g}{g},k'=0)$ (suppose that $\frac{n-g}{g}$ is an integer)
	have strictly less storage overhead and less repair bandwidth
	than that of $(n,k,g)$ optimal-LRC code, under the same fault-tolerant capability.
\end{theorem}
\begin{proof}
When $s=\frac{n-g}{g}$, the storage overhead of $\mathcal{C}(n,k+g,\frac{n-g}{g},k'=0)$ is
\begin{align*}
\frac{(s+1)\cdot n}{s\cdot(k+g)}=\frac{(\frac{n-g}{g}+1)\cdot n}{(\frac{n-g}{g})\cdot (k+g)}=\frac{n^2}{(n-g)\cdot(k+g)},
\end{align*}
while the storage overhead of $(n,k,g)$ optimal-LRC is $\frac{n}{k}$. We have that
\begin{align*}
\frac{n^2}{(n-g)\cdot(k+g)}<\frac{n}{k}\Leftrightarrow nk<(n-g)(k+g)\Leftrightarrow k+g<n.
\end{align*}
The last inequality is obviously true. Therefore, $\mathcal{C}(n,k+g,\frac{n-g}{g},k'=0)$
have strictly less storage overhead than that of $(n,k,g)$ optimal-LRC.

By assumption, we have that $2g>n-k+1$, then we can obtain that
\begin{align*}
	(s-1)(r+1)+1=(\frac{n-g}{g}-1)(n-k-g+1)+1<(\frac{n-g}{g}-1) g+1=n-2g+1<k+g.
\end{align*}
The condition in Theorem \ref{th5} is satisfied, and the fault-tolerant of 
$\mathcal{C}(n,k+g,\frac{n-g}{g},k'=0)$ is $r+1=n-k-g+1$, which is equal to the 
fault-tolerant of $(n,k,g)$ optimal-LRC code.

Recall that the average repair bandwidth ratio of all nodes of $(n,k,g)$ optimal-LRC 
and $\mathcal{C}(n,k+g,\frac{n-g}{g},k'=0)$ is $\frac{n-g}{k\cdot g}$ and 
$\frac{s+1}{k+g}$, respectively. We have that
\begin{align*}
\frac{s+1}{k+g}<\frac{n-g}{k\cdot g}\Leftrightarrow\frac{(\frac{n-g}{g})+1}{k+g}<\frac{n-g}{k\cdot g}\Leftrightarrow nk<(n-g)(k+g)\Leftrightarrow k+g<n.
\end{align*}
The last inequality is obviously true. Therefore, $\mathcal{C}(n,k+g,\frac{n-g}{g},k'=0)$
have strictly less repair bandwidth than that of $(n,k,g)$ optimal-LRC.
\end{proof}

\begin{figure}
	\centering
	\includegraphics[width=0.50\linewidth]{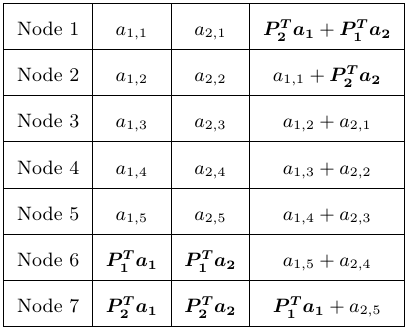}
	\caption{An example of code $\mathcal{C}(n=7,k=5,s=2,k'=0)$.}
	\label{fig.7}
\end{figure}

Consider the code $\mathcal{C}(n=7,k=5,s=2,k'=0)$ which is shown in Fig. \ref{fig.7}.
Suppose that node 1 fails, we can first download 2 symbols
$\mathbf{P}_{2}^T\mathbf{a}_1,\mathbf{P}_{1}^T\mathbf{a}_2$ to obtain the piggyback function $p_1=\mathbf{P}_{2}^T\mathbf{a}_1+\mathbf{P}_{1}^T\mathbf{a}_2$, then download 2 symbols $a_{1,1}+\mathbf{P}_{2}^T\mathbf{a}_2, \mathbf{P}_{2}^T\mathbf{a}_2$ to recover $a_{1,1}$ and finally download 2 symbols $a_{1,2}+a_{2,1},a_{1,2}$ to recover $a_{2,1}$. The repair bandwidth of node 1 is 6 symbols.
Similarly, the repair bandwidth of any single-node failure
is 6 symbols.

We claim that the code $\mathcal{C}(n=7,k=5,s=2,k'=0)$ can tolerant any $r+1=3$ node failures. Suppose that nodes 2, 4 and 6 fail, we have that the number of surviving nodes between two failed nodes are $t_1=1$, $t_2=1$ and $t_3=2$. We can first repair the symbol $\mathbf{P}_{1}^T\mathbf{a}_2$ in node 6 by downloading symbols $\mathbf{P}_{2}^T\mathbf{a}_1+\mathbf{P}_{1}^T\mathbf{a}_2, \mathbf{P}_{2}^T\mathbf{a}_1$. Since the symbol $\mathbf{P}_{1}^T\mathbf{a}_2$ has been recovered, according to the MDS property of the second column, we can recover $a_{2,2},a_{2,4}$. Then, we can download symbols $a_{2,5},\mathbf{P}_{1}^T\mathbf{a}_1+a_{2,5}$ to recover $\mathbf{P}_{1}^T\mathbf{a}_1$. Finally, since $\mathbf{P}_{1}^T\mathbf{a}_1$ has been repaired, we can recover $a_{1,2},a_{1,4}$ by the MDS property of the first column. Up to now, we have recovered the first two symbols in the three failed nodes. Since the third symbol in each of the three failed nodes is a piggyback function, we can recover the symbol by reading some data symbols. The repair method of any three failed nodes is similar as the above repair method.

\section{Comparison}
\label{sec:com}
In this section, we evaluate the repair bandwidth for our piggybacking codes
$\mathcal{C}(n,k,s,k')$ and other related codes,
such as existing piggybacking codes \cite{2017Piggybacking,2021piggyback}
and $(n,k,g)$ Azure-LRC code \cite{huang2012} under the same fault-tolerance and
the storage overhead.

\subsection{Codes $\mathcal{C}(n,k,s,k'=k)$ VS Piggybacking Codes}
\label{sec:com1}
Recall that OOP codes
\cite{2019AnEfficient} have the lowest repair bandwidth for any single-node failure
among the existing piggybacking codes when $r\leq10$ and sub-packetization is $r-1+\lfloor \sqrt{r-1}\rfloor$ or $r-1+\lceil \sqrt{r-1}\rceil$,
the codes in \cite{2021piggybacking} have the lowest repair bandwidth for any
single-node failure among the existing piggybacking codes when $r\geq10$ and sub-packetization is {$r$}.
REPB codes \cite{2018Repair} have small repair bandwidth for any single-data-node
with sub-packetization {usually} less than $r$.
Moreover, the codes in \cite{2021piggyback} have lower repair bandwidth than the existing
piggybacking codes when {$r\geq10$} and sub-packetization is less than $r$.
There are two constructions in \cite{2021piggyback}, the first construction
in \cite{2021piggyback} has larger repair bandwidth than the second construction.
We choose codes in \cite{2018Repair,2021piggybacking} and the second construction
in \cite{2021piggyback} as the main comparison.

Let $\mathcal{C}_{1}$ be the second piggybacking codes in \cite{2021piggyback}
and $\mathcal{C}_{2}$ be the codes in \cite{2021piggybacking}. Fig. \ref{fig.8}
shows the average repair bandwidth ratio of all nodes for codes $\mathcal{C}_{1}$,
$\mathcal{C}_{2}$, REPB codes \cite{2018Repair} and the proposed codes $\mathcal{C}(n,k,s,k'=k)$,
where $r=8,9$ and $k=10,11,\ldots,100$.
Note that the sub-packetization of REPB codes and $\mathcal{C}_{1}$ is inexplicit
and {usually} less than $r$. In Fig. \ref{fig.8}, we choose the lower bound of the repair bandwidth of REPB codes and $\mathcal{C}_{1}$.
The sub-packetization of $\mathcal{C}_{2}$ is $r$.

The results in Fig. \ref{fig.8} demonstrate that the proposed codes $\mathcal{C}(n,k,s,k'=k)$
have the lowest average repair bandwidth ratio of all nodes compared to the existing
piggybacking codes, when the sub-packetization level is less than $r$ and $k\geq 30$. Note that
the sub-packetization of our codes $\mathcal{C}(n,k,s,k'=k)$ is $\sqrt{r}$, which is
lower than that of $\mathcal{C}_{2}$.

\begin{figure}[htpb]
	\centering
	\subfigure[$r=8$]{               %小图题的名称
		\includegraphics[width=7cm]{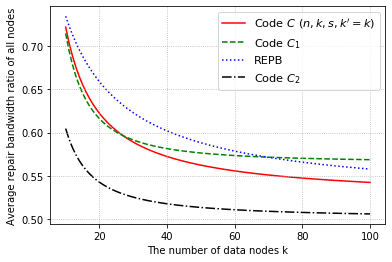}}
	%\hspace{0in}
	\subfigure[$r=9$]{
		\includegraphics[width=7cm]{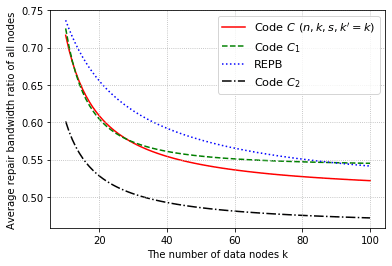}}
	\caption{Average repair bandwidth ratio of all nodes for codes $\mathcal{C}(n,k,s,k'=k)$, REPB, $\mathcal{C}_{1}$ and $\mathcal{C}_{2}$, where $r=8$ and $k=10,11,\ldots,100$ in $(a)$, $r=9$ and $k=10,11,\ldots,100$ in $(b)$.}
	\label{fig.8}
\end{figure}

\subsection{Codes $\mathcal{C}(n,k,s,k'=k-sr-1)$ VS Piggybacking Codes}
Recall that OOP codes \cite{2019AnEfficient} have the lowest
repair bandwidth for any single-node failure among the existing
piggybacking codes when $r\leq10$, where the sub-packetization of OOP codes is
$r-1+\lfloor \sqrt{r-1}\rfloor$ or $r-1+\lceil \sqrt{r-1}\rceil$, and the minimum value
of average repair bandwidth ratio of all nodes $\gamma_{OOP}^{all}$ is
\begin{eqnarray}
&&\gamma_{OOP}^{all}=\frac{1}{k+r}(k\cdot \frac{2\sqrt{r-1}+1}{2\sqrt{r-1}+r}+r\cdot(\frac{\sqrt{r-1}}{r}+\frac{1}{r}+\frac{(r-1)^2-\sqrt{(r-1)^3}}{kr})).\label{eq0}
\end{eqnarray}

In the following, we show that our codes $\mathcal{C}(n,k,s,k'=k-sr-1)$ have strictly
less repair bandwidth than OOP codes and thus have less repair bandwidth than all the existing
piggybacking codes when $r\leq10$.

\begin{lemma}
	\label{col1}
Let $\gamma_{1}^{all}$ be the average repair bandwidth ratio of all nodes of
$\mathcal{C}(n,k,s,k'=k-sr-1)$. When $2\leq r\ll k,2+\sqrt{r-1}\leq s$, we have
	\begin{eqnarray}
		&&\underset{k\rightarrow +\infty}{lim}\gamma_{1}^{all}< \underset{k\rightarrow +\infty}{lim}\gamma_{OOP}^{all}.\nonumber
	\end{eqnarray}
\end{lemma}
\begin{proof}
According to Theorem. \ref{th1}, since $k-k'=h=sr+1$, we have
			\begin{eqnarray}
			\underset{k\rightarrow +\infty}{lim}\gamma_{1}^{all}&\leq &\frac{s^2}{(sr+r)(s+1)}+\frac{1}{s+1}\nonumber\\&=&\frac{s^2}{(s+1)^2}\cdot\frac{1}{r}+\frac{1}{s+1}<\frac{1}{r}+\frac{1}{s+1}.\nonumber
		\end{eqnarray}
		From Eq. \eqref{eq0}, when $r\geq2$, we have
		\begin{eqnarray}
			(\underset{k\rightarrow+\infty}{lim}\gamma_{OOP}^{all})-\frac{1}{r}
			&=&\frac{2\sqrt{r-1}+1}{2\sqrt{r-1}+r}-\frac{1}{r}\nonumber\\
			&=&\frac{2\sqrt{r-1}}{2\sqrt{r-1}+r}\cdot\frac{r-1}{r}\nonumber\\
			&\geq&\frac{2\sqrt{r-1}}{2\sqrt{r-1}+r}\cdot\frac{1}{2}\nonumber\\
			&=&\frac{\sqrt{r-1}}{2\sqrt{r-1}+r}.\label{eq5}
		\end{eqnarray}
		When $s\geq\sqrt{r-1}+2$ and $r\geq2$, we have
		\begin{eqnarray}
			s&\geq&\sqrt{r-1}+2\nonumber\\
			&=&(\sqrt{r-1}+1)+1\nonumber\\
			&=&(\frac{r-1}{\sqrt{r-1}}+1)+1\nonumber\\
			&\geq&(\frac{r-1}{\sqrt{r-1}}+1)+\frac{1}{\sqrt{r-1}}\nonumber\\
			&=&\frac{r+\sqrt{r-1}}{\sqrt{r-1}}.\label{eq6}
		\end{eqnarray}
	 From Eq. \eqref{eq6} and Eq. \eqref{eq5}, we have
	 \begin{eqnarray}
	 	\frac{1}{s+1}&\leq&(\frac{r+\sqrt{r-1}}{\sqrt{r-1}}+1)^{-1}\nonumber\\
	 	&=&\frac{\sqrt{r-1}}{2\sqrt{r-1}+r}\nonumber\\
	 	&\leq&(\underset{k\rightarrow+\infty}{lim}\gamma_{OOP}^{all})-\frac{1}{r}.\nonumber
	 \end{eqnarray}
 	Therefore, we have
 	\begin{eqnarray}
 		&&\underset{k\rightarrow +\infty}{lim}\gamma_{1}^{all}<\frac{1}{r}+\frac{1}{s+1}\leq \underset{k\rightarrow +\infty}{lim}\gamma_{OOP}^{all}.\nonumber
 	\end{eqnarray}
	\end{proof}
It is easy to see that the storage overhead of $\mathcal{C}(n,k,s,k'=k-sr-1)$
ranges from $\frac{k+r}{k}$ to $\frac{k}{k-r}\cdot\frac{k+r}{k}$.
According to Lemma \ref{col1}, our codes $\mathcal{C}(n,k,s,k'=k-sr+1)$ have strictly
less repair bandwidth than all the existing piggybacking codes when $r\ll k$, at a cost
of slightly more storage overhead.

\subsection{The proposed Codes VS LRC}
\label{sec:com2}
Next, we evaluate the repair bandwidth of our codes $\mathcal{C}(n,k+g,\frac{n}{g},k'=0)$, codes $\mathcal{C}(n,k+g,\frac{n-g}{g},k'=0)$,
Azure-LRC \cite{huang2012} and optimal-LRC. According to Theorem \ref{th6}, the repair bandwidth
of our $\mathcal{C}(n,k+g,\frac{n}{g},k'=0)$ is strictly less than that of $(n,k,g)$
Azure-LRC \cite{huang2012}. Fig. \ref{fig.9} shows the average repair bandwidth ratio of all nodes
for $\mathcal{C}(n,k+g,\frac{n}{g},k'=0)$ and $(n,k,g)$
Azure-LRC when the fault-tolerance is 8 and $n=100$. The results demonstrate
that $\mathcal{C}(n,k+g,\frac{n}{g},k'=0)$ have strictly less repair bandwidth than
$(n,k,g)$ Azure-LRC. Moreover, our
$\mathcal{C}(n,k+g,\frac{n}{g},k'=0)$ have less storage overhead than
$(n,k,g)$ Azure-LRC.
For example, $\mathcal{C}(n,k+g,\frac{n}{g},k'=0)$
have 44.92\% less repair bandwidth and 6.16\% less storage overhead than $(n,k,g)$ Azure-LRC
when $(n,k,g)=(100,73,20)$.

\begin{figure}[htpb]
	\centering
	\includegraphics[width=0.50\linewidth]{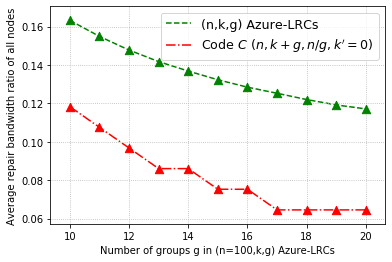}
	\caption{Average repair bandwidth ratio of all nodes for codes $\mathcal{C}(n,k+g,\frac{n}{g},k'=0)$ and $(n,k,g)$ Azure-LRC, where $10\leq g\leq20$, fault-tolerance is 8, and $n=100$.}
	\label{fig.9}
\end{figure}

\begin{figure}
	\centering
	\includegraphics[width=0.50\linewidth]{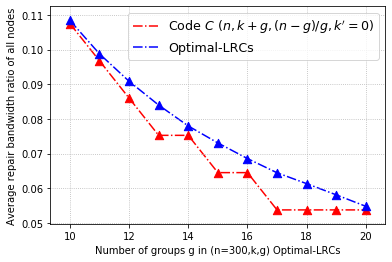}
	\caption{Average repair bandwidth ratio of all nodes for codes $\mathcal{C}(n,k+g,\frac{n-g}{g},k'=0)$ and $(n,k,g)$ optimal-LRC under the same fault-tolerance and code length $n=100$, where $10\leq g\leq20$.}
	\label{fig.10}
\end{figure}
According to Theorem \ref{th7}, the repair bandwidth
of our $\mathcal{C}(n,k+g,\frac{n-g}{g},k'=0)$ is strictly less than that of $(n,k,g)$
optimal-LRC.
Fig. \ref{fig.10} shows the average repair bandwidth ratio of all nodes
for $\mathcal{C}(n,k+g,\frac{n-g}{g},k'=0)$ and $(n,k,g)$
optimal-LRC  under the same fault-tolerance and $n=100$. The results demonstrate
that $\mathcal{C}(n,k+g,\frac{n-g}{g},k'=0)$ have strictly less repair bandwidth than
$(n,k,g)$ optimal-LRC.

\section{Conclusion}
\label{sec:con}
In this paper, we propose two new piggybacking coding designs. We propose
one class of piggybacking codes based on the first design
that are MDS codes, have lower repair bandwidth than the
existing piggybacking codes when $r\geq8$ and the sub-packetization is $\alpha<r$.
We also propose another piggybacking codes based on the second design that are non-MDS codes
and have better tradeoff between
storage overhead and repair bandwidth, compared with
Azure-LRC and optimal-LRC for some parameters.
One future work is to generalize the piggybacking coding design over codewords
of more than two different MDS codes.
Another future work is to obtain the condition of $\mathcal{C}(n,k,s,k'=0)$
that can recover any $r+2$ failures.

\ifCLASSOPTIONcaptionsoff
  \newpage
\fi

\bibliographystyle{IEEEtran}
\bibliography{CNC-v1}
\end{document}